\documentclass[journal,draftclsnofoot,onecolumn,12pt]{IEEEtran}
\IEEEoverridecommandlockouts
\usepackage{cite}
\usepackage{amsmath,amssymb,amsfonts}
\usepackage{float}
\usepackage{graphicx}
\usepackage{caption}
\usepackage{textcomp}
\usepackage{algorithm}
\usepackage{amsthm}
\usepackage{algpseudocode}
\usepackage{adjustbox}
\usepackage{xcolor}
\usepackage{verbatim}

\DeclareMathOperator*{\argmax}{arg\,max}
\newtheorem{proposition}[]{Proposition}

  \allowdisplaybreaks
\begin{document}

\title{Trajectory and Resource Optimization for UAV Synthetic Aperture Radar\\
\thanks{This work was supported in part by the Deutsche Forschungsgemeinschaft (DFG, German Research Foundation) GRK 2680 – Project-ID 437847244.}
}
\author{\IEEEauthorblockN{Mohamed-Amine~Lahmeri\IEEEauthorrefmark{1},
Walid Ghanem\IEEEauthorrefmark{1}, Christina Knill\IEEEauthorrefmark{2}, and
Robert Schober\IEEEauthorrefmark{1}}\\
\IEEEauthorblockA{\IEEEauthorrefmark{1}Friedrich-Alexander-Universit\"at Erlangen-N\"urnberg, Germany\\
\IEEEauthorrefmark{2}Ulm University, Ulm, Germany\\
Emails:\{amine.lahmeri, walid.ghanem, robert.schober\}@fau.de,
christina.knill@uni-ulm.de}}

\maketitle

\begin{abstract}
In this paper, we study the trajectory and resource optimization for lightweight rotary-wing unmanned aerial vehicles (UAVs) equipped with a synthetic aperture radar (SAR) system. The UAV's mission is to perform SAR imaging of a given area of interest (AoI). In this setup, real-time communication with a base station (BS) is required to facilitate live mission planning for the drone. For this purpose, a non-convex mixed-integer non-linear program (MINLP) is formulated such that the UAV resources and three-dimensional (3D) trajectory are jointly optimized for maximization of the drone radar ground coverage. We present a low-complexity sub-optimal algorithm based on successive convex approximation (SCA) for solving the problem, and perform a finite search to optimize the total distance traversed by the UAV for maximal coverage. We show that the proposed 3D trajectory planning achieves at least 70\% improvement in radar ground coverage compared to benchmark schemes employing constant powers for communication or radar imaging. We also show that positioning the BS near the AoI can significantly improve the radar coverage of the UAV.
\end{abstract}
\vspace{-3mm}
\section{Introduction}
Since their invention in 1922, RAdio Detection And Ranging (radar) systems have been widely used in numerous applications, including detection, tracking, and forecasting\cite{b00}. 
For imaging, radar systems are known to be resilient to weather conditions and to be able to operate day and night, which is not the case for conventional optical imaging, where clouds can block the view and lack of light significantly degrades image quality\cite{b0}. \par
With the current growth of the drone industry, the use of unmanned aerial vehicles (UAVs) is becoming easier than ever. This is confirmed by the widespread use of commercial drones in a plethora of applications such as aerial photography, delivery, surveillance, and communication\cite{b01}. Although small, lightweight UAVs are limited by their battery capacity \cite{b44}, their flexibility has motivated the deployment of synthetic aperture radar (SAR) systems onboard. Despite the fact that spaceborne SAR systems have been used onboard high-altitude platforms for more than 30 years \cite{b0}, their deployment onboard small, lightweight commercial UAVs is a relatively new idea that is currently undergoing testing and study.\par 
For UAV-SAR systems, the main research focus has been on radar signal processing \cite{b15}, while only few results are available for trajectory design. Experimentally, the trajectories are designed manually or automatically using waypoint coverage planning. Unlike classical trajectory planning problems for drones, for UAV-SAR systems, the trajectory design depends on the operating mode of the radar. The most commonly used trajectories are circular trajectories and linear trajectories for stripmap mode SAR. For instance, back and forth linear motion has been used in \cite{b4}, while circular SAR was used in \cite{b6}. Path planning concepts for rough areas were proposed in \cite{b12} and a multi-objective optimization problem for minimization of the distance traveled by the UAV was formulated and solved in \cite{b13}. However, in these works, communication with a ground base station (BS) was not accounted for and a geosynchronous UAV (GEO-UAV) was considered, where the area was illuminated by a GEO-SAR satellite. Thus, the path planning problem was different from the one considered in this paper. \par  
Although UAVs could avoid transmitting raw radar data to the ground by storing it onboard and then processing it offline, for live mission planning real-time communication with a BS is needed to allow the UAV to make adjustments based on the radar data.\footnote{Adjustments may include the redefinition of the UAV AoI, and restarting or stopping the current mission of the UAV-SAR after evaluating the processed radar image.} Enabling reliable communication  between drones and BSs has been widely studied in the general UAV literature but not for SAR systems \cite{b1,b2}. One of the few works that investigated the transmission of radar data to a ground BS was \cite{b11}, where the free-space path loss model was used to characterize the air-to-ground channel. However, in \cite{b11}, a bi-static passive UAV-SAR system was considered, which is different from the SAR system studied in this work.  The authors of \cite{b5} used a GHz wireless local area network (WLAN) to connect a UAV-SAR system with the ground. However, in \cite{b11,b5}, UAV trajectory design and resources optimization were not considered.\par
In this paper, we develop a resource allocation scheme for the optimization of the radar ground coverage of a UAV-SAR system performing terrain surface imaging. To allow live mission planning, real-time communication with the ground is assumed. Motivated by Boustrephedon coverage planning\cite{b9}, the SAR system operates in the stripmap mode and the UAV adopts a linear-shaped trajectory. In this framework, the three-dimensional (3D) trajectory and the resources of the UAV are jointly optimized to achieve maximum radar ground coverage. Our main contributions can be summarized as follows: 
\begin{itemize}
    \item We investigate trajectory and resource allocation optimization for a low-altitude UAV-SAR system, where raw radar data is transmitted to the ground in real-time, for the first time in the literature.
    \item We present a novel framework for modelling the SAR and communication subsystems and exploit it to formulate an optimization problem for drone coverage maximization.  	
    \item  We propose a low-complexity sub-optimal solution for the formulated non-convex  mixed-integer non-linear program  (MINLP) based on successive convex approximation (SCA).
    \item  We demonstrate that the proposed solution outperforms three benchmark schemes in terms of radar ground coverage.  We also show that placing the BS near the area of interest (AoI) significantly improves the radar coverage of the 
    UAV-SAR system.
\end{itemize}
\textit{Notations}: 
In this paper, lower-case letters $x$ refer to scalar numbers, while boldface lower-case letters $\mathbf{x}$ denote vectors and boldface upper-case letters $\mathbf{X}$ refer to matrices.  $\{a, ..., b\}$ denotes the set of all integers between $a$ and $b$ and $|\cdot|$ denotes the absolute value operator. For given sets $\mathcal{X}_1$ and $\mathcal{X}_2$, the set $\mathcal{X}_1 \setminus \mathcal{X}_2$  is the set that comprises all elements of $\mathcal{X}_1$ that are not in $\mathcal{X}_2$.  $\mathbb{N}$ represents the set of natural numbers while $\mathbb{R}^{N \times 1}$ represents the set of all $N \times 1$ vectors with real values entries. For a vector $\mathbf{x}\in\mathbb{R}^{N \times 1} $, $\mathbf{x}^T$ stands for the transpose of the vector. For a real matrix of size $N$,  $ \mathbf{A} \in \mathbb{R}^{N\times N}$, $\mathbf{A}\succeq \mathbf{0} $ indicates that $\mathbf{A}$ is a positive semi-definite matrix.  For a real-valued function $f(x)$, $f'(x)$ denotes the derivative of $f$.
\section{System Model}
We consider a small rotary-wing UAV equipped with a side-looking radar antenna designed to perform SAR. The main mission of the UAV is to identify and scan a certain AoI on the ground. Due to the computational constraints of the UAV and the need for live mission monitoring, the UAV transmits data to the ground in real time. The data comprises the raw radar data as well as other information needed to process the radar signal such as synchronization and localization data.
\begin{figure}[]
    \centering
    \captionsetup{justification=centering,margin=2cm}
    \includegraphics[width=6in]{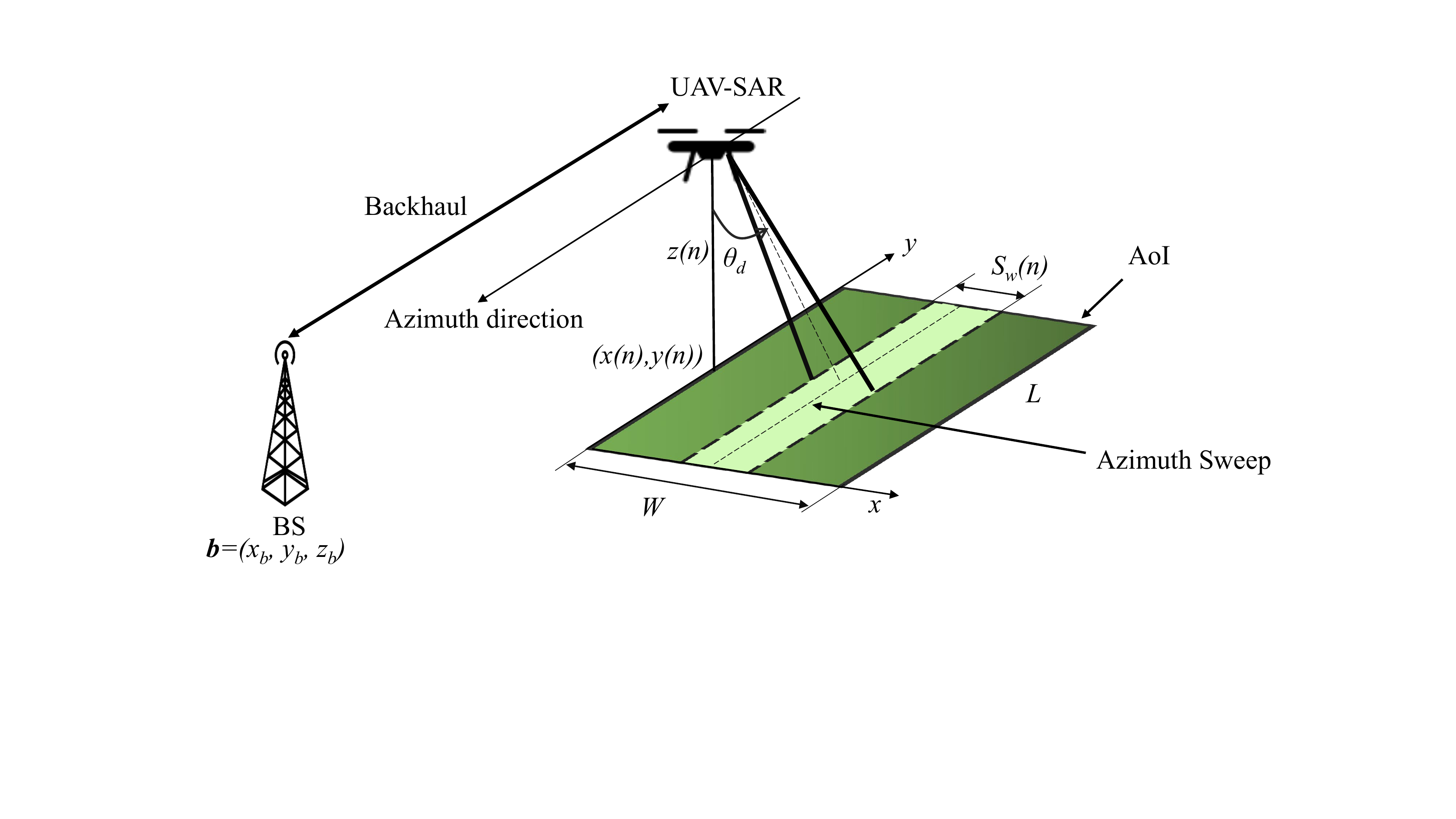}
    \caption{System model for UAV-SAR system with a backhaul communication link.}
    \label{fig:SystemModel}
\end{figure}
\subsection{Trajectory Design}
\label{Sec:TrajectoryDesign}
We assume that the AoIs are rectangular with given length $L$.  Without loss of generality, we place the origin of the adopted coordinate system at that corner of the AoI, where the UAV mission will start, see Fig. \ref{fig:SystemModel}.  We suppose that the drone  flies at constant velocity $v$. The $y$-axis represents the azimuth direction while the $x$-axis defines the range direction, as shown in Fig. \ref{fig:SystemModel}. 
We discretize time such that the length of the area scanned in $y$-direction in time slot $n$ of duration $\delta_t$, denoted by $\Delta_y$, is given by
\begin{equation}
   \Delta_y =\delta_t v .
\end{equation}
The UAV trajectory consists of a set of azimuth sweeps\footnote{An azimuth sweep is the trajectory of the drone in azimuth direction from one edge of the AoI to the other, i.e., from  $y=0$ to $y=L$ and vice versa, see Fig. \ref{fig:SystemModel}.}  where successive ground radar scans are adjacent such that there are no gaps in the coverage. To mathematically characterize the linear trajectory of the UAV-SAR, we define a set of time slots  $\mathcal{A}$ as follows:
\begin{equation}
    \mathcal{A}=\{ k\in \{1,...,NM\} |  (k-1) \;\mathrm{mod} \; M \neq 0 \},
\end{equation}
where $N$ denotes the total number of the drone azimuth sweeps, while $M$ is the number of time slots for each drone azimuth sweep.  $\Delta_y$ is set such that $M=\frac{L}{\Delta_y}$ is an integer. Consequently,  set $\mathcal{A}^c$ is defined as follows:
\begin{equation}
    \mathcal{A}^c=\{1,...,NM\}\setminus \mathcal{A}= \{1, M+1, 2M+1,...,(N-1)M+1 \}.
\end{equation}
The drone location in time slot $n \in \{1,...,NM\}$ is defined by $\mathbf{u}(n)=(x(n),y(n),z(n)) $, where $x(n)$ and $y(n)$ represent the position in the $xy$-plane, while $z(n)$ denotes the drone altitude. The UAV makes its turns when it reaches the edge of the AoI at the end of an azimuth sweep. Then, it moves to the $y$-position for the next azimuth sweep as defined later in (\ref{eq:turn}). 
To account for the limitations of practical radar systems, sets $\mathcal{A}$ and $\mathcal{A}^c$ will be used to (i) force a back-and-forth linear motion and (ii) set all of the radar parameters, such as the radar power and altitude, constant during a given sweep \cite{b5,b6}. Set $\mathcal{A}^c$ contains the indices of the first time slots at the beginning of the azimuth sweeps where the radar parameters can be adjusted, whereas set $\mathcal{A}$ contains the indices of all other time slots, where the parameters are fixed. The following constraints on the drone trajectory are imposed: 
\begin{gather}
    \mathrm{C1}: x(1)=-\tan(\theta_1)z(1),\\\label{eq:turn}
\mathrm{C2}: x(n)=x(n-1)+z(n-1)\tan(\theta_2)-z(n)\tan(\theta_1), \forall n \in \mathcal{A}^c\setminus\{1\},\\
    \mathrm{C3}: x(n+1)=x(n), \forall n \in \mathcal{A},\\
    \mathrm{C4}: z(n+1)=z(n), \forall n \in \mathcal{A}.
\end{gather}
Here, angles $\theta_1$ and $\theta_2$ are given by: 
\begin{gather}
    \theta_1= \theta_d-\frac{\Theta_{3\mathrm{dB}}}{2},\\
        \theta_2=\theta_d+\frac{\Theta_{3\mathrm{dB}}}{2},
\end{gather}
where $\Theta_{3\mathrm{dB}}$ is the radar antenna 3dB  beamwidth and $\theta_d$ is the radar depression angle.
Constraint $\mathrm{C2}$ ensures that two successive radar scans of the ground generated by two successive drone sweeps are adjacent without gap. Constraint $\mathrm{C3}$ imposes a linear motion, while constraint $\mathrm{C4}$ ensures a fixed altitude during each drone sweep. Forcing a linear motion for the stripmap UAV-SAR does not mean that the trajectory is fixed. In fact, at each turn, the drone altitude $z$ and the $x$-positions can be adjusted. Only the $y$-positions are fixed such that back and forth stripmap mode motion is guaranteed. The $y$-positions are defined as follows: 
\begin{equation}
y(1)=0,  \; y(n+1)=y(n) + c(n) \Delta_y, \hspace{1mm } \forall n \in \{1,...,NM-1\},\\
\end{equation}where vector $\mathbf{c} =[c(1),..., c(NM) ]^T \in \mathbb{R}^{NM\times 1}$ is defined as follows: 
\begin{equation}
    \mathbf{c}= [ \overbrace{\underbrace{1,...,1}_\text{$M$ times},\underbrace{-1,...,-1,}_\text{$M$ times}...]^T}^\text{$N\times M$ times},
\end{equation}

Considering these constraints, the shape of a typical trajectory is shown in Fig. \ref{fig:trajectory}.
\subsection{Radar Imaging}

Since we operate at low altitudes, the maximal imageable area on the ground is limited by the antenna beamwidth. Therefore, the radar swath width in time slot $n$, $S_w(n)$, is given as follows: 
\begin{equation}
    S_w(n)=z(n) \left( \tan(\theta_2)-\tan(\theta_1)\right).
\end{equation}
 Consequently, the infinitesimal area covered in time slot $n$ is $\Delta_y S_w(n)$ and the total ground area covered by the drone, denoted by $C$, is the sum of all infinitesimal coverage areas: 
\begin{equation}
    C=\sum_{n=1}^{NM} \Delta_y S_w(n)=\sum_{n=1}^{NM} \Delta_y z(n) \left( \tan(\theta_2)-\tan(\theta_1)\right).
\end{equation}
The data rate produced by SAR sensing, in the time slot $n$, can be expressed as \cite{b7}:
\begin{equation}
    R_{\mathrm{min}}(n)= \;B_r \left( \frac{2\;z(n)\;\Omega}{c}+\tau_p \right) \mathrm{PRF}\hspace{2mm} [\mathrm{bit/s}],
\end{equation}
where $B_r$ is the radar bandwidth, $\tau_p$ is the radar pulse duration, $c$ is the speed of light, PRF is the radar pulse repetition frequency, and $\Omega$ is given by: 
\begin{equation}
\Omega=\frac{\cos(\theta_1)-\cos(\theta_2)}{\cos(\theta_1)\cos(\theta_2)}.
\end{equation}
Furthermore, the radar signal-to-noise ratio (SNR) is given as follows \cite{b8}:
\begin{equation}
    \mathrm{SNR}(n)=\frac{P_{\mathrm{sar}}(n)\; G_t\; G_r \;\lambda^3 \;\sigma_0 \;c \;\tau_p \mathrm{PRF}\;\sin^2(\theta_d)}{(4\pi)^4 \;z^3(n) \;k \;T_o \;NF \;\;B_r \;L_{\mathrm{tot}} \;v },
\end{equation}
where $P_{\mathrm{sar}}(n)$ is the radar transmit power in time slot $n$,  $G_t$ and $G_r$ are the radar antenna gains for transmission and reception, respectively,  $\lambda$ is the radar wavelength, $T_o$ is the equivalent noise temperature, $k$ is Boltzmann's constant, $L_{\mathrm{tot}}$
represents the combined losses, $NF$ is the system noise figure, and $\sigma_o$ is the backscattering coefficient. 
\subsection{UAV Data Link}
We denote the location of the BS by $\mathbf{b}=(x_b,y_b,z_b)$. Consequently, the distance between the drone and the BS is given by:
\begin{equation}
d(n)= ||\mathbf{u}(n)-\mathbf{b} ||_2=\sqrt{(x(n)-x_b)^2+(y(n)-y_b)^2+(z(n)-z_b)^2}.
\end{equation}
We suppose that the UAV  altitude is sufficiently high to allow obstacle-free communication with the BS over a Line-of-Sight (LoS) link. Thus, based on the free space path loss model, the instantaneous backhaul throughput from the UAV to the BS is given by: 
\begin{equation}
    R(n)= B_{c} \; \log_2\left(1+\frac{P_{\mathrm{com}}(n) \;\gamma}{d^2(n)}\right),
\end{equation}
where $P_{\mathrm{com}}(n)$ is the power allocated for  communication in time slot $n$, $B_{c}$ is the communication bandwidth, and $\gamma$ is the reference channel gain divided by the noise variance. 
\\To guarantee successful real-time communication between the drone and the BS, the following inequality must be satisfied at all times: 
\begin{gather}
   \mathrm{C8}:  R(n) \geq R_{\mathrm{min}}(n)+R_{\mathrm{sl}}, \forall \; n,
\end{gather}
where $R_{\mathrm{sl}}$ is the rate needed to transmit the synchronization and localization data  necessary for successful radar image processing.

\section{Resource Allocation and Trajectory Optimization Framework}
In this section, we present the problem formulation for maximizing the radar ground coverage and then explain our proposed solution to the problem. 
\subsection{Problem Formulation}
Our objective is the maximization of the drone ground coverage $C$ while satisfying the constraints imposed by the radar and communication subsystems. Thus, we optimize the drone trajectory $
\{ \mathbf{z},\mathbf{x},N\}$ and power allocation $\{\mathbf{q}, \mathbf{P_{\mathrm{com}}}, \mathbf{P}_{\mathrm{sar}}\}$, where vectors $\mathbf{z}$, $\mathbf{x}$, $\mathbf{q}$, $ \mathbf{P}_{\mathrm{com}}$, and $\mathbf{P}_{\mathrm{sar}}$,  represent the collections of all $z(n), x(n), q(n), P_{\rm com}(n), P_{\rm sar}(n), \forall n$, respectively, and $q(n)$ is the energy available in the drone's battery in time slot $n$. To this end, the following optimization problem is formulated:
\begin{subequations}
\begin{alignat*}{2}
&\!(\mathrm{P.1}):\max_{\mathbf{z},\mathbf{x},\mathbf{q}, \mathbf{P_{\mathrm{sar}}}, \mathbf{P_{\mathrm{com}}},N} \hspace{3mm} C       & \qquad&  \\
&\text{s.t.} \hspace{3mm}\mathrm{C1},\mathrm{C2},\mathrm{C3},\mathrm{C4},\mathrm{C8} &   & \; \\
&  \mathrm{C5:}P_{\mathrm{sar}}(n+1)=P_{\mathrm{sar}}(n),\hspace{1mm}  \forall n \in \mathcal{A},                &      &  \\ & \mathrm{C6}: \;z_{\mathrm{min}}\leq z(n) \leq z_{\mathrm{max}}, \hspace{1mm} \forall n,                 &      &  \\&  \mathrm{C7}:\mathrm{SNR}(n) \geq \mathrm{SNR_{min}}, \hspace{1mm} \forall n,                &      &     
 \\&    \mathrm{C9}: \;0\leq P_{\mathrm{sar}}(n)\leq P_{\mathrm{sar}}^{\mathrm{max}}, 0\leq P_{\mathrm{com}}(n)\leq P_{\mathrm{com}}^{\mathrm{max}} ,\hspace{1mm} \forall n,              &      & \\
& \mathrm{C10}: q(1)= q_{\mathrm{start}},\;q(n) \geq 0, \hspace{1mm} \forall n,     & & \\
& \mathrm{C11}: q(n+1) = q(n) - \delta_t\left(P_{\mathrm{com}}(n)+P_{\mathrm{sar}}(n)+P_{\mathrm{prop}} \right), \; \forall n.     & &  
\end{alignat*}
\end{subequations}
Constraints $\mathrm{C1}$, $\mathrm{C2}$, $\mathrm{C3}$, and $\mathrm{C4}$ define the shape of the UAV trajectory as described in Section \ref{Sec:TrajectoryDesign}. Constraint $\mathrm{C5}$ fixes the radar power along each azimuth sweep.  Constraint $\mathrm{C6}$ determines the allowed range for the UAV's operational altitude. Constraint $\mathrm{C7}$ ensures the minimum SNR,  $\mathrm{SNR_{min}}$, required for SAR imaging is achieved. Constraint $\mathrm{C8}$ guarantees successful real-time transmission of the UAV's onboard data to the BS. Constraint $\mathrm{C9}$ ensures that the radar and communication transmit power are non-negative and do not exceed the maximum permissible levels, respectively. Constraints $\mathrm{C10}$ and $\mathrm{C11}$ specify the energy level available at the start of the mission, $q_{\mathrm{start}}$, and in time slot $n$, respectively. Here, the propulsion power required for maintaining the operation of the drone, denoted by $P_{\mathrm{prop}}$, is assumed to be constant since the drone velocity is fixed.\par
Problem (P.1) is a non-convex MINLP  due to the non-convex constraint C8 and integer variable $N \in \mathbb{N}$, which does not only affect the objective function but also the dimension of all other optimization variables. Therefore, $N$ cannot be optimized jointly with the remaining variables. Furthermore, non-convex MINLPs are usually difficult to solve. Thus, in the next section, we propose an SCA based sub-optimal algorithm for the considered problem. 
\subsection{Solution of Optimization Problem}
We start by fixing variable $N$ and denote problem $\mathrm{(P.1)}$ for a given $N$ by $\mathrm{\widetilde{(P.1)}}$:
\begin{subequations}
\begin{alignat*}{2}
&\!\mathrm{\widetilde{(P.1)}}:\max_{\mathbf{z},\mathbf{x},\mathbf{q}, \mathbf{P_{\mathrm{sar}}}, \mathbf{P_{\mathrm{com}}}} \hspace{3mm} C       & \qquad&  \\
&\text{s.t.} \hspace{3mm}\mathrm{C1-C11}. &   & \; 
\end{alignat*}
\end{subequations}
However, even for a given $N$, problem $\mathrm{\widetilde{(P.1)}}$ is non-convex and difficult to solve. Thus, we use SCA to determine a low-complexity sub-optimal solution for problem $\mathrm{\widetilde{(P.1)}}$ in an iterative manner. Then, we prove that the set of feasible $N$ is finite and  perform a search for the optimal $N^*$. In what follows, we discuss the steps needed to approximate problem $\mathrm{\widetilde{(P.1)}}$ as a convex optimization problem. \par 
As a first step, we  transform constraint C7, which involves a cubic function of drone altitude $z$, using second-order cone programming. To this end, we introduce  slack variables $\psi(n)$ stacked in vector $\boldsymbol{\psi}=[\psi(1),...,\psi(NM)]^T\in \mathbb{R}^{NM\times 1}$, constant $\beta$,  and reformulate C7 as follows \cite{b10}: 
\begin{gather}
    \exists \psi(n) \geq 0 \in \mathbb{R}, \widetilde{\mathrm{C7a}}:\begin{bmatrix}
\psi(n) & z(n) \\
z(n) & 1 
\end{bmatrix} \succeq \mathbf{0},\forall n,\\ \widetilde{\mathrm{C7b}}: \begin{bmatrix}
P_{\rm sar}(n) \beta & \psi(n) \\
\psi(n) & z(n) 
\end{bmatrix} \succeq \mathbf{0},   \forall n ,
\end{gather}
where $\beta=\frac{ G_t\; G_r\;\lambda^3 \;\sigma_0 \;c \;\tau_p \mathrm{PRF}\;\sin^2(\theta_d)}{(4\pi)^4 \;k \;T_o \;NF \;\;B_r \;L_{\mathrm{tot}} \;v\; \mathrm{SNR_{min}} }$. Furthermore, constraint C8 can be transformed, simplified, and equivalently rewritten as follows: 
\begin{equation}
\label{eq:C8}
     {\rm C8}: \left( A2^{\alpha z(n)} -1\right)\left((x(n)-x_b)^2+(z(n)-z_b)^2+(y(n)-y_b)^2 \right)\leq P_{\mathrm{com}}(n)\gamma, \forall n,
\end{equation}
where $A=2^{\frac{B_r\;\tau_p\;\mathrm{PRF} }{B_c}}$ and $\alpha=\frac{2 \;\Omega \;B_r\; \mathrm{PRF}}{c \;B_c}$.
For a given $y(n)$ and positive $A$ and $\alpha$, the term $\left( A2^{\alpha z(n)} -1\right) \allowbreak(y(n)-y_b)^2$ in C8 is convex with respect to $z(n)$, since:
\begin{equation}
\frac{\partial ^2}{\partial z^2(n)}\left( A2^{\alpha z(n)} -1\right)(y(n)-y_b)^2= \alpha^2 A2^{\alpha z(n)} \log^2(2)(y(n)-y_b)^2 > 0.
\end{equation}
Based on its Hessian matrix, it can be shown that term $ h_1(z(n), x(n))=\left( A2^{\alpha z(n)} -1\right)(x(n)-x_b)^2$ in C8 is not convex. Thus, we rewrite function $ h_1( z(n), x(n))$ as follows:
\begin{equation}
    h_1\left(z(n),x(n)\right)=\frac{1}{2}\left( \left(f_1(z(n))+f_2(x(n))\right)^2-f_1(z(n))^2-f_2(x(n))^2\right),\label{eq: ProductLinearization}
\end{equation}
where $f_1(z(n))=(A2^{\alpha z(n)} -1)$ and $f_2(x(n))=(x(n)-x_b)^2$ are both convex functions. In (\ref{eq: ProductLinearization}), function $h_1(z(n),x(n))$ is written as a difference of convex functions. We use  Taylor approximation to obtain convex lower bounds for $F_1(z(n))=f_1^2(z(n))$ and $F_2(x(n))=f_2^2(x(n))$ around points $z^{(j)}(n)$ and $x^{(j)}(n)$:
\begin{gather}
    F_1(z(n)) \geq F_1(z^{(j)}(n))+ F_1'(z^{(j)}(n))(z(n)-z^{(j)}(n))\label{eq:Taylor1}\\
    F_2(x(n)) \geq F_2(x^{(j)}(n))+ F_2'(x^{(j)}(n))(x(n)-x^{(j)}(n)),\label{eq:Taylor2}
\end{gather}
where $F_1'(z(n))=A\alpha \ln(2) 2^{\alpha z(n)+1}( 2^{\alpha z(n)}-1)$ and $ F_2'(x(n))=4(x(n)-x_b)^3$. Similarly, we rewrite the non-convex term $ h_2(z(n))= (A2^{\alpha z(n)}-1)(z(n)-z_b)^2$  as follows: 
\begin{equation}
        h_2(z(n))=\frac{1}{2}\left((f_3(z(n))+f_4(z(n)))^2-f_3(z(n))^2-f_4(z(n))^2\right),\label{eq:ProductLinea2}
\end{equation}
where convex lower bounds for functions $F_3(z(n))=f_3^2(z(n))=(A2^{\alpha z(n)}-1)^2$ and $F_4(z(n))=f_4^2(z(n))=(z(n)-z_b)^2$ are given by:
\begin{gather}
    F_3(z(n)) \geq F_3(z^{(j)}(n))+  F_3'(z^{(j)}(n))(z(n)-z^{(j)}(n)),\label{eq:Taylor3}\\
    F_4(z(n)) \geq F_4(z^{(j)}(n))+  F_4'(z^{(j)}(n))(z(n)-z^{(j)}(n)),\label{eq:Taylor4}
\end{gather}
where $ F_3'(z(n))=A\alpha \ln(2) 2^{\alpha z(n)+1}( 2^{\alpha z(n)}-1)$ and $  F_4'(z(n))=4(z(n)-z_b)^3$. The last term in C8, i.e., $ \left( A2^{\alpha z(n)} -1\right) (y(n)-y_b)$, is convex and no further transformation is required. We use (\ref{eq:Taylor1}) and (\ref{eq:Taylor2}) to provide a lower bound for (\ref{eq: ProductLinearization}) and (\ref{eq:Taylor3}) and (\ref{eq:Taylor4}) to provide a lower bound for (\ref{eq:ProductLinea2}). Then, we combine all terms and rewrite constraint C8 in a compact form as follows: 
\begin{gather}
     \sum_{i \in \{1,3\}}(f_i(\delta_i)+f_{i+1}(\delta_{i+1}))^2+\sum_{i=1}^4 F_i(\delta_i^{(j)}) +  F_i'(\delta_i^{(j)})(\delta_i-\delta_i^{(j)})\notag\\
    +(A2^{\alpha z(n)} -1) (y(n)-y_b)\leq P_{\mathrm{com}}(n)\gamma, \forall n,
 \end{gather}
 where $\delta_2=x(n)$  and $\delta_i=z(n)$, $\forall i \neq 2$.
 In order to successfully solve the problem, we introduce new slack vectors $\mathbf{t}=[t(1),...,t(NM)]^T\in \mathbb{R}^{NM\times 1}$, $\mathbf{u}=[u(1),...,u(NM)]^T\in \mathbb{R}^{NM\times 1}$, and two new constraints $\mathrm{\widetilde{C8a}}$ and $\mathrm{\widetilde{C8b}}$, and reformulate constraint C8 as follows: 
\begin{gather}
       \mathrm{\widetilde{C8 }}: t(n)^2+u(n)^2+\sum_{i=1}^4 F_i(\delta_i^{(j)}) +  F_i'(\delta_i^{(j)})(\delta_i^{(j)}-\delta_i)\notag\\
    +(A2^{\alpha z(n)} -1) (y(n)-y_b)\leq P_{\mathrm{com}}(n)\gamma, \forall n,\\
         \mathrm{\widetilde{C8a}}:f_1(z(n))+f_2(x(n))\leq t(n), \forall n,\\
   \mathrm{\widetilde{C8b}}: f_3(z(n))+f_4(z(n))\leq   u(n), \forall n.
\end{gather}
Now, problem $\mathrm{\widetilde{(P.1)}}$ can be approximated with the following convex optimization problem:
\begin{subequations}
\begin{alignat*}{2}
&\!\mathrm{(P.2)}:\max_{\mathbf{z},\mathbf{x},\mathbf{q}, \mathbf{P_{\mathrm{sar}}}, \mathbf{P_{\mathrm{com}}}} \hspace{3mm} C       & \qquad&  \\
&\text{s.t.} \hspace{3mm}\mathrm{C1-C6,\widetilde{C7a},\widetilde{C7b},\widetilde{C8},\widetilde{C8a},\widetilde{C8b}, C9-C11}. &   & \; 
\end{alignat*}
\end{subequations}
\begin{algorithm}
  \caption{Successive Convex Approximation}\label{euclid}
  \begin{algorithmic}[1] 
       \label{algorithm1} \State Set initial point, $(\mathbf{z}^{(1)},\mathbf{x}^{(1)})$, iteration index $j=1$, and error tolerance $0< \epsilon \ll 1$.
      \State \textbf{repeat}
             \State $  $Get coverage $ C^{(j)}, \hspace{1mm} \mathbf{z}, $ and $ \mathbf{x} $ by solving (P.2) for the point $ (\mathbf{z}^{(j)},\mathbf{x}^{(j)}).  $           \State$ $Set $ j=j+1,\hspace{1mm} \mathbf{ z}^{(j)}= \mathbf{z},\hspace{1mm} \mathbf{ x}^{(j)}=\mathbf{x}.$ 
            \State \textbf{until} $\big |\frac{C^{(j)}-C^{(j-1)}}{C^{(j)}}\big|\leq \epsilon$

       \State \textbf{return} solution $\mathbf{z}^*=\mathbf{z}^{(j)}$, $\mathbf{x}^*=\mathbf{x}^{(j)}$

  \end{algorithmic}
\end{algorithm}
 To solve $\mathrm{\widetilde{(P.1)}}$, we use SCA such that in each iteration $j$ of \textbf{Algorithm 1}, we  solve problem (P.2) and update the solution $\{\mathbf{z}^{(j)},\mathbf{x}^{(j)},\mathbf{q}^{(j)},\mathbf{P}_{\mathrm{sar}}^{(j)},\mathbf{P}_{\rm com}^{(j)}\}$. Following  \cite{convergence}, it can be shown hat \textbf{Algorithm 1} converges to a locally optimal solution and has polynomial time computational complexity. Problem (P.2) can be efficiently solved by classical convex program solvers such as CVX \cite{b10}. \par
\subsection{Optimal Number of Drone Sweeps $N^*$}
To solve Problem (P.1), we first solve problem $\mathrm{\widetilde{(P.1)}}$ for every fixed and feasible number of drone sweeps $N$, then we select the minimum number of azimuth sweeps that results in the maximum coverage to save drone energy.  In the following proposition, we show that the number of feasible $N$ is finite so that we can perform this exhaustive search. 
\begin{proposition}
\label{proposition}
The feasible number of azimuth sweeps $N$ is finite.
\end{proposition}
\begin{proof}
Please refer to Appendix. \ref{SecondAppendix}.
\end{proof}

\section{Simulation and Discussion}
\begin{table}[]
\centering
\caption{System parameters \cite{table1,table2,table3}. }
\begin{tabular}{|c|c|c|c|}
\hline
Parameter                & Value & Parameter    & Value   \\ \hline
$M$                       & 120   & $f$            & 2 Ghz   \\ \hline
$L$    & 50 m  & $\beta$         & $10^4$ W$^{-1}$ m$^3 $        \\ \hline
$z_{\rm max}$               & 100 m  & $\tau_p$        & 1 $\mu$ s        \\ \hline
$z_{\rm min}$                     & 5 m      & $B_r=B_c$         &   100 MHz      \\ \hline
$\Delta_y$ & 0.5 m  &   $\gamma$      &  20 dB\\ \hline
$\theta_d$ & 45$^o $     &   $q_{\rm start}$      & 100 Wh   \\ \hline
$\theta_{\rm 3db}$                  & 30$^o$     &$P_{\rm prop}$  & 140 W   \\ \hline
SNR$_{\rm min}$             & 20 dB &$v$    & 5 m/s    \\ \hline
PRF                      & 100 Hz      &$P_{\rm sar}^{\rm max}$         &  46 dBm   \\ \hline
$P_{\rm com}^{\rm max}$& 40 dBm &$R_{\rm sl}$ &1 kbit\\ \hline
\end{tabular}
\label{Tab:System}
\end{table}
\begin{figure}[]
    \centering
    \includegraphics[width=4in]{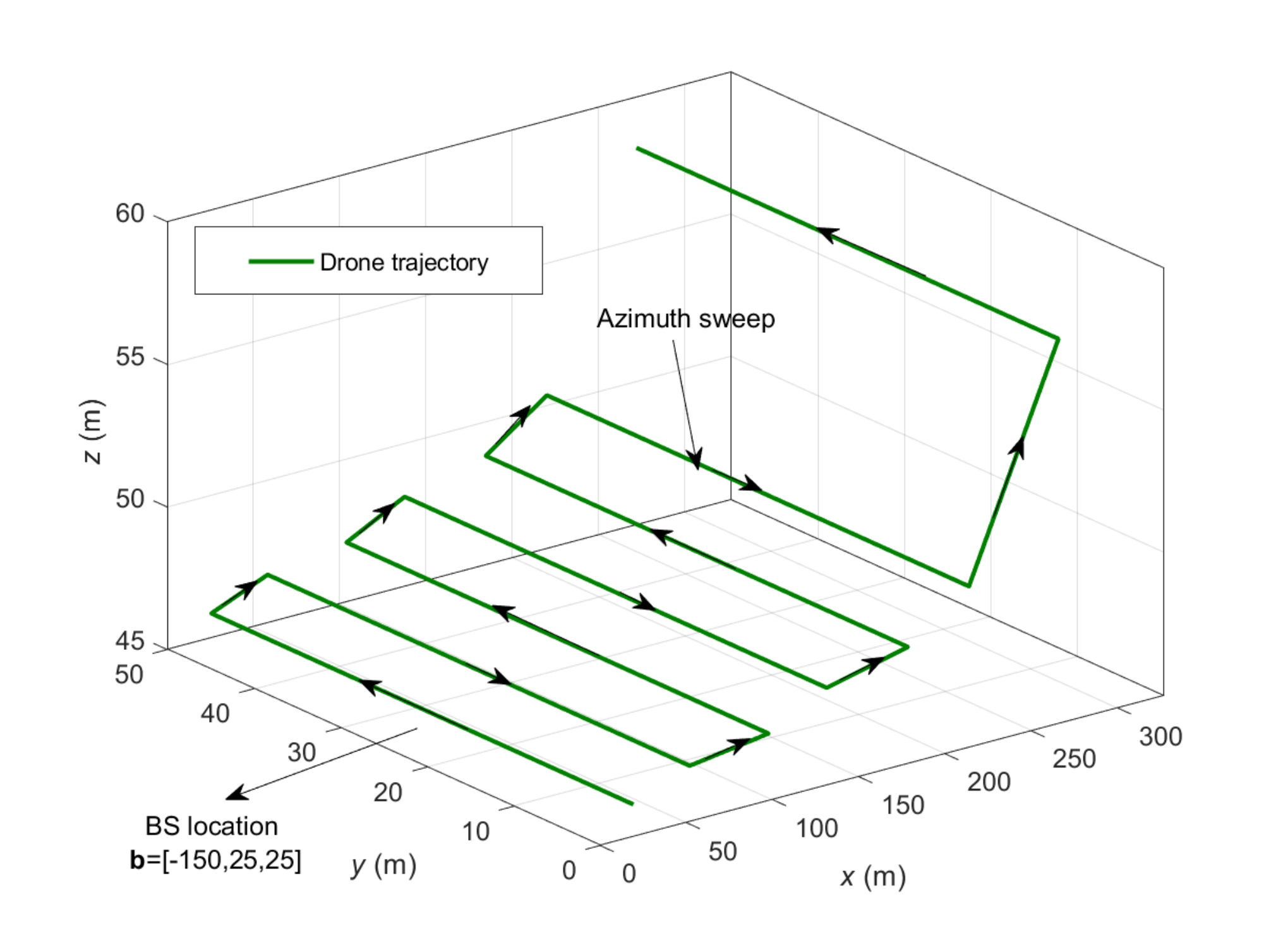}
    \caption{Optimized 3D UAV trajectory for stripmap mode UAV-SAR.}
    \label{fig:trajectory}
\end{figure}
In this section, we present simulation results for the 3D UAV trajectory, the resource allocation, and the impact of the BS placement on the UAV-SAR performance. We also compare the proposed solution with different benchmark schemes. \par
\subsection{Benchmark Schemes}
As relevant benchmark schemes are not available in the literature, we compare the proposed scheme to the following benchmark schemes:
\begin{itemize}
    \item \textbf{Benchmark scheme 1:} For this scheme, we solve (P.1) using a fixed but optimized drone altitude in all time slots.  
        \item \textbf{Benchmark scheme 2:} For this scheme, we fix the power for communication $P_{\rm com}$ in all time slots and optimize it along with the other optimization variables by solving a modified version of (P.1).
    \item \textbf{Benchmark scheme 3:}  For this scheme, we fix the UAV-SAR radar transmit power $P_{\rm sar}$ for all drone azimuth sweeps and solve again a modified version of (P.1).
\end{itemize}
\subsection{Simulation Results }
 Figure \ref{fig:trajectory} and Figure \ref{fig:Coverage} show the drone trajectory and the radar ground coverage, respectively, whereas Figure \ref{fig:CommunicationPower}  and Figure \ref{fig:RadarPower} provide insights into the UAV-SAR resource allocation. The adopted simulation parameters are summarized in Table \ref{Tab:System}.\par
For Figure \ref{fig:trajectory}, the BS is located at $\mathbf{b}=[-150, 25, 25]$ m, and we plot the trajectory of the drone over the AoI. The optimal number of drone sweeps is $N^*=7$. The shape of the drone's trajectory shows the back and forth motion imposed by constraints C1-C4. The change in the drone altitude  can be explained by the relationship between the $x$ and $z$ coordinates. Flying at high altitude results in a wider swath width and better coverage but also imposes large UAV turns. For a given number of UAV turns, flying at a high altitude at the start of the mission results in moving quickly far away from the BS, where successful communication with the BS may not be possible as throughput $R(n)$ decreases with the UAV-BS distance. Thus, to ensure successful backhaul communication, the UAV starts flying at a relatively low altitude (46 m) and increases its altitude considerably at the end of the mission (58 m) to maximize the radar coverage, i.e., the objective function of (P.1). We note that, for the case depicted in Figure \ref{fig:trajectory}, the coverage is not limited by the total available energy $q_{\rm start}$, but by the UAV-SAR's ability to meet constraint C8 as it moves far away from the BS.\par
\begin{figure}[]
    \centering
    \includegraphics[width=4in]{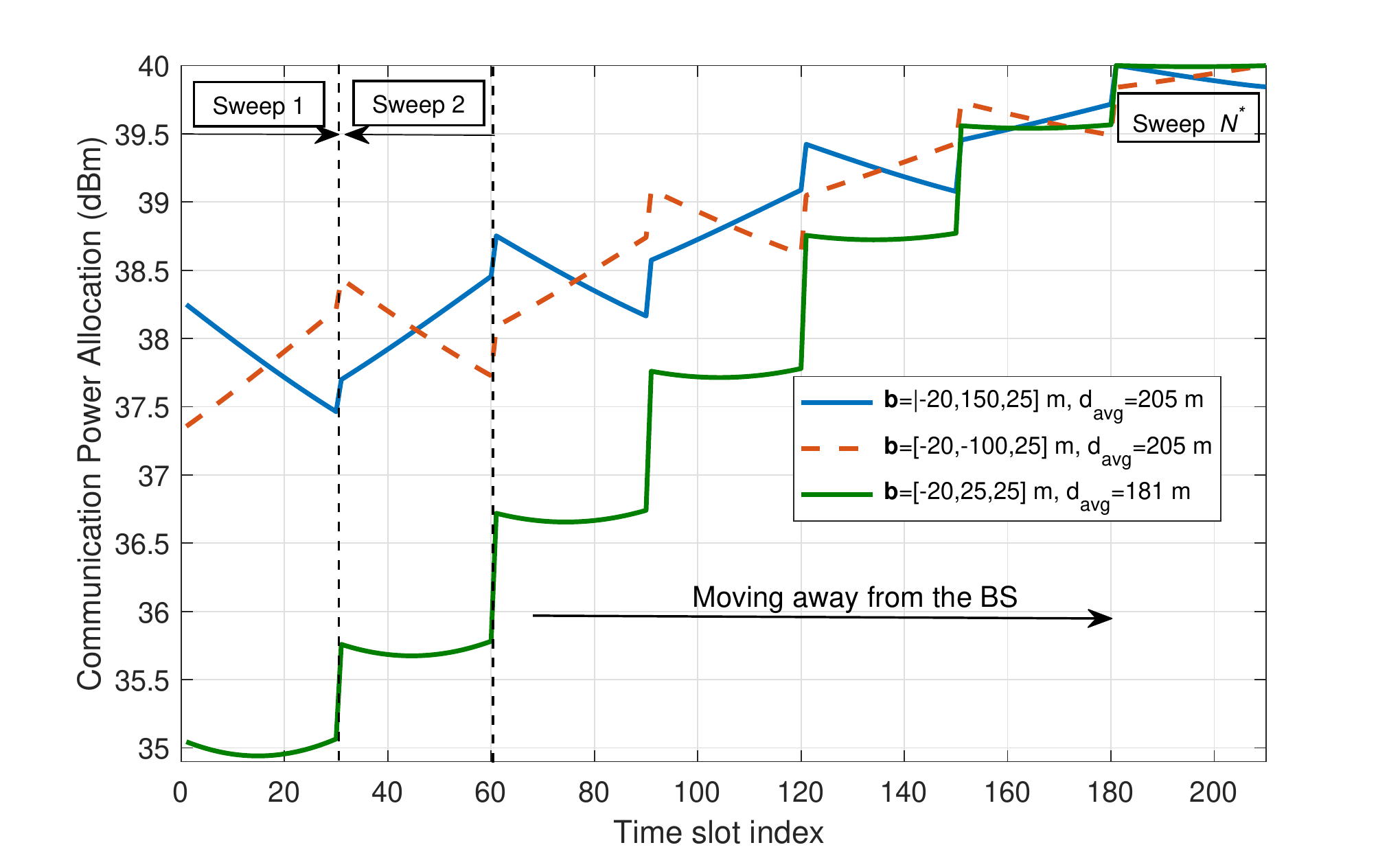}
    \caption{Communication power allocation.}
    \label{fig:CommunicationPower}
\end{figure}
In Figure \ref{fig:CommunicationPower}, we show the power allocated for communication as a function of time for different locations of the BS. As the drone moves away from the BS, more power is required for communication to compensate for the increasing path loss such that C8 can be met. The green solid line corresponds to the case where the BS is located at the line crossing the middle of the AoI.  Therefore, the power allocated is almost constant during each azimuth sweep. This is because there is no significant change in the UAV-BS distance during an azimuth sweep. When the BS is placed to the left or to the right of the AoI, represented by the solid blue line and the dashed red line, respectively, the power slowly increases during sweeps where the drone moves away from the BS, and then in the following sweeps, it slowly decreases as the UAV moves towards the BS.  For these cases, the average communication power consumption is higher than for the case when the BS is located in the middle. This is because the average distance to the BS, defined as $d_{\mathrm{avg}}=\frac{1}{NM}\sum\limits_{n=1}^{NM} d(n)$, is smaller when the BS is located in the middle of the AoI, i.e.  at $\mathbf{b}=[-20,25,25]$ m. We therefore conclude that the positioning of the BS relative to the AoI has a significant impact on the power consumed for communication.\par
\begin{figure}[]
    \centering
    \includegraphics[width=4in]{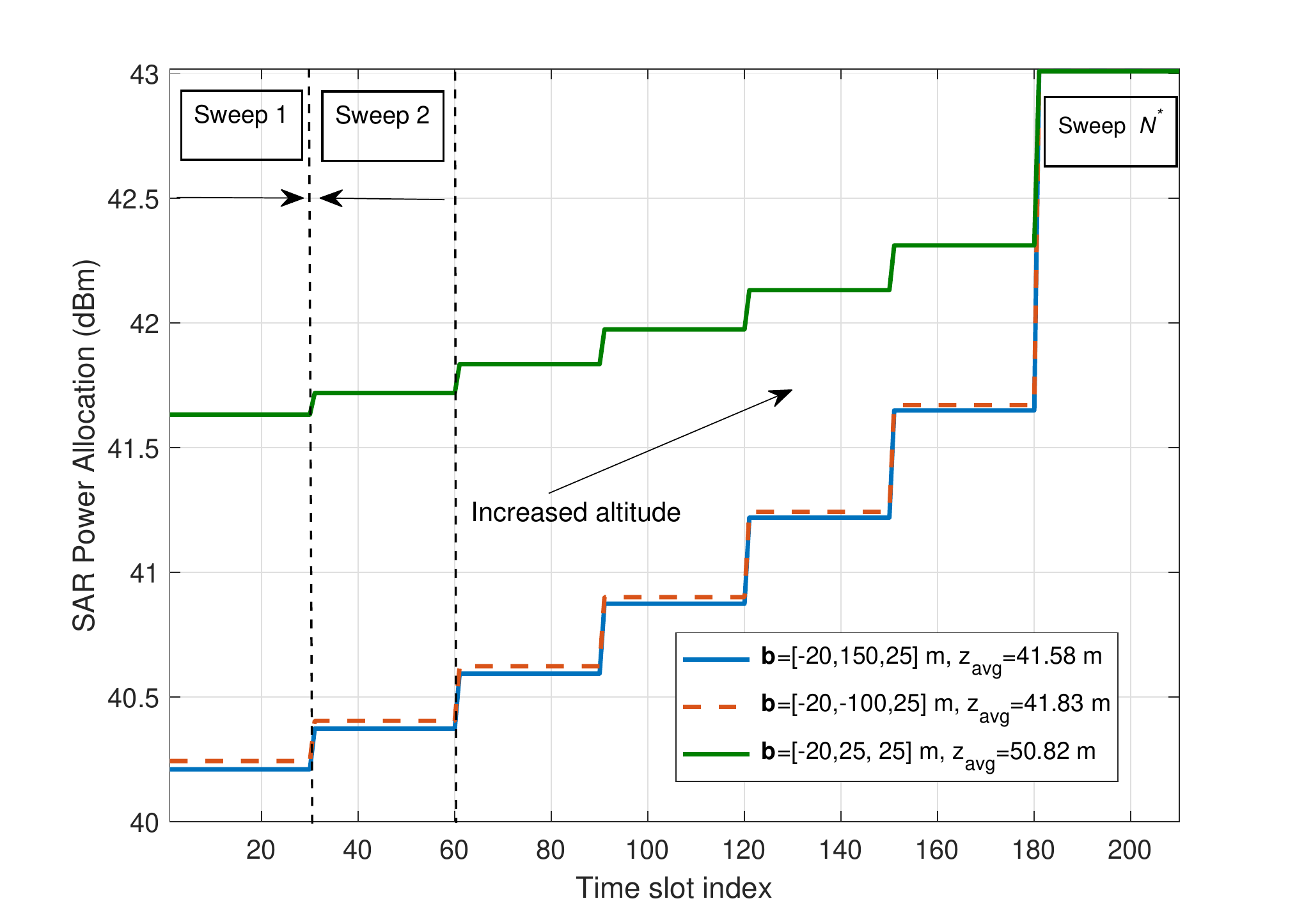}
    \caption{Radar power allocation.}
    \label{fig:RadarPower}
\end{figure}
In Figure \ref{fig:RadarPower}, we depict the radar power allocation as a function of the time slot index. Unlike the communication power profile, the radar power is not impacted by the distance to the BS and only increases with the altitude of the UAV as the minimum SNR required for SAR has to be satisfied. Thus, the radar power during each azimuth scan is constant. The power allocated to the SAR system increases from one sweep to the next as the drone altitude continuously increases, cf. Figure  \ref{fig:trajectory}.  On average, the power allocated for the radar system is higher when the BS is placed at $[-20, 25, 25]$ m since the average flying altitude, defined by $z_{\mathrm{avg}}=\frac{1}{NM}\sum\limits_{n=1}^{NM} z(n)$, is higher than for the two other BS placement.\par

In Figure \ref{fig:Coverage}, we plot the coverage obtained with \textbf{Algorithm 1} as a function of the number of sweeps $N$ for different BS placements. This figure confirms the previous conclusion that the distance to the BS  significantly affects the coverage of the drones. For example, when the BS is placed at $\mathbf{b}=[-200, 80, 25]$ m, the maximum coverage is approximately $0.4\times 10^{4}\;\mathrm{m}^2$ whereas when the BS is placed closer to the AoI, i.e., at $\mathbf{b}=[-150, 80, 25] $ m, the maximum radar coverage increases to $ 1.8 \times 10^4 \; \mathrm{m}^2$. Furthermore, we observe that the coverage increases with the number of sweeps until the optimal number of azimuth sweeps $N^*$ is reached. Then, the radar coverage starts to decrease  and finally drops to zero. This is due to the limited battery capacity of the drone, $q_{\rm start}$. Problem (P.1) becomes infeasible if $N$ is chosen too large, see Appendix \ref{SecondAppendix}.   Figure \ref{fig:Coverage} also reveals that the proposed scheme outperforms  the considered benchmark schemes. The area covered with the proposed scheme is more than twice as large as that for benchmark scheme 2. More importantly, the proposed solution achieves this larger coverage with fewer sweeps $(N^*=8)$ compared to benchmark scheme 2 $(N^*=19)$, hence the drone has to spend less energy.  We also observe a 70\% improvement in radar coverage compared to benchmark scheme 3. This improvement was achieved at the cost of total 9 sweeps ($N^*=9)$ compared to fewer sweeps ($N^*=5)$ for benchmark scheme 3 . Moreover, the proposed 3D trajectory provides a 1000 m$^2$ improvement  in radar coverage compared to benchmark scheme 1 employing a constant UAV altitude.

\begin{figure}[]
    \centering
    \includegraphics[width=4in]{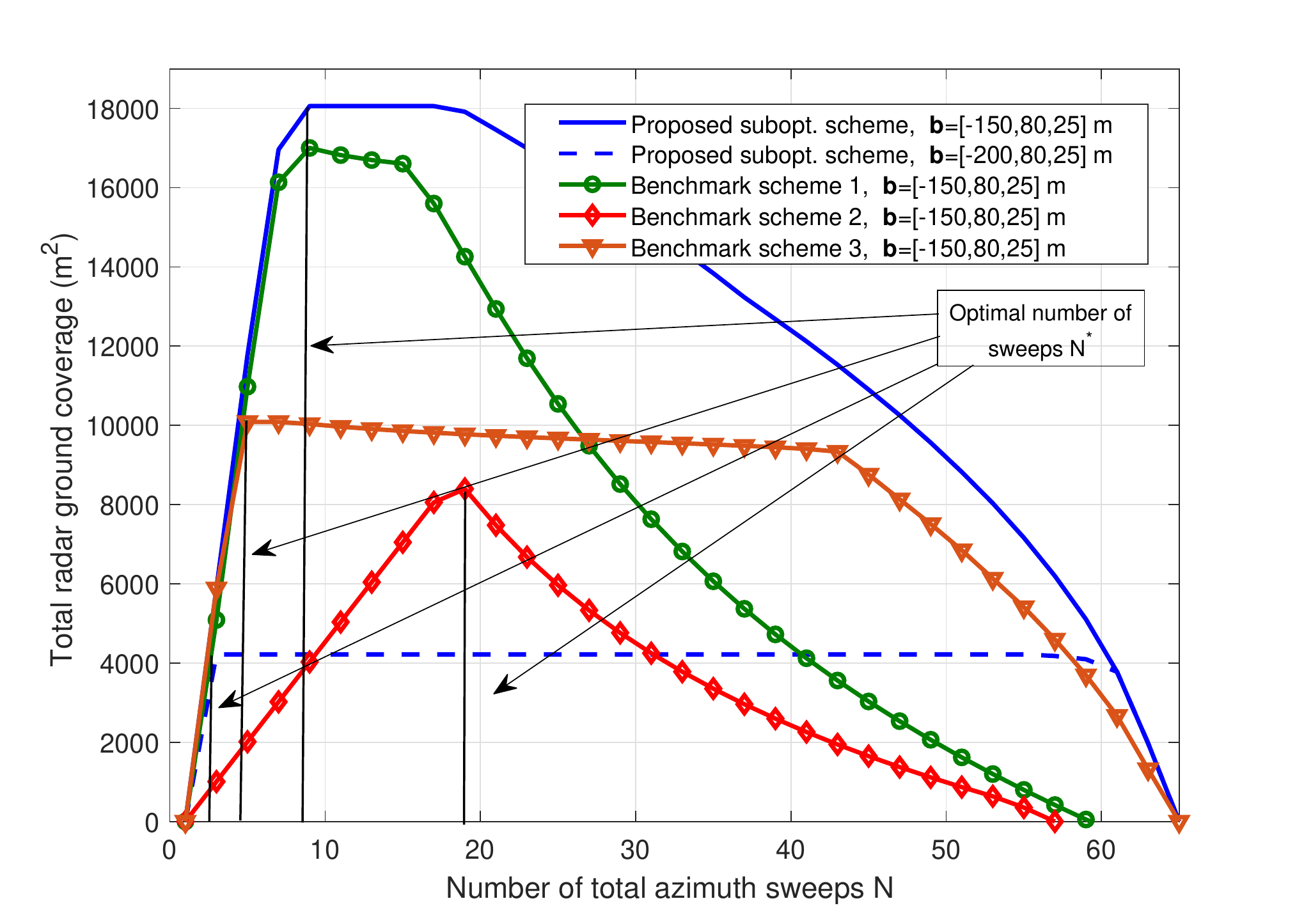}
    \caption{Coverage as function of number of turns $N$.}
    \label{fig:Coverage}
\end{figure}

\vspace{-3mm}
\section{Conclusion}
In this paper, we provided a joint trajectory and resource allocation optimization framework for UAV-SAR systems.  We consider real-time data transmission of the radar data to the ground to facilitate live mission planning. The resource allocation algorithm was formulated as non-convex optimization problem for maximization of the UAV radar coverage. We propose a low-complexity sub-optimal solution based on SCA followed by a search for the optimal number of azimuth sweeps to minimize the drone trajectory length. Our numerical simulations reveal that the proposed 3D path planning scheme clearly outperforms three benchmark schemes with fixed communication power, radar transmit power, and UAV altitude, respectively. Furthermore, our results also show that placing the BS near the AoI is beneficial for maximizing the drone radar coverage.

  \appendices
  \section{ Proof of Proposition \ref{proposition}}
  \label{SecondAppendix}
  In this appendix, we prove that the set of feasible $N$ for (P.1) is upper bounded. We start by reformulating constraint C11 as follows: 
  \begin{equation}
     \mathrm{C11}: \delta_t\left(P_{\mathrm{com}}(n)+P_{\mathrm{sar}}(n)+P_{\mathrm{prop}} \right)=q(n) -q(n+1) , \; \forall n. \label{eq:C11}
  \end{equation}
   We use constraint C9 to obtain a lower bound on the left-hand side of (\ref{eq:C11}):
\begin{equation}
\delta_t P_{\mathrm{prop}} \leq q(n) -q(n+1).
  \end{equation}

  Then, we apply a summation to both sides of (\ref{eq:C11}) as follows:
    \begin{equation}
     \sum_{n=1}^{NM-1}\delta_tP_{\mathrm{prop}} \leq \sum_{n=1}^{NM-1}q(n) -q(n+1).
  \end{equation}
We can  rigorously show the following result by induction, which corresponds to the sum of a telescoping series: 
      \begin{equation}
       \sum_{n=1}^{NM-1}q(n) -q(n+1)=q(1)-q(NM). \label{eq:telescop}
        \end{equation}
        Now, we insert (\ref{eq:telescop}) into (\ref{eq:C11}):
        \begin{equation}
           \sum_{n=1}^{NM-1}\delta_t P_{\mathrm{prop}} \leq q_{\mathrm{start}}-q(NM) \leq q_{\rm start} .
        \end{equation}
Since the drone flies at a constant velocity $v$, $P_{\rm prop}$ is constant. Thus, we obtain an upper bound on $N$ as follows: 
        
        \begin{equation}
   N \leq \frac{1}{M} \left(\frac{q_{\mathrm{start}}}{\delta_t P_{\mathrm{prop}}}+1\right). \label{eq:upperboundN}
\end{equation}
Problem (P.2) is infeasible for $N$ exceeding the upper bound provided in (\ref{eq:upperboundN}) as this would violate constraint C11.

\end{document}